\documentclass[11pt]{amsart}

\usepackage{new-macros}

\usepackage[utf8]{inputenc}
\usepackage[T1]{fontenc}
\usepackage[backend=biber,
            isbn=false,
            doi=false,
            maxbibnames=5,
            style=alphabetic,
            citestyle=alphabetic]{biblatex}
\usepackage{url}
\renewbibmacro{in:}{}
\bibliography{heterotic.bib}

\usepackage{xcolor}
\definecolor{e-mail}{rgb}{0,.40,.80}
\definecolor{reference}{rgb}{.40,.60,.2}
\definecolor{citation}{rgb}{0,.25,.5}

\def\Td{{\rm Td}}

\def\wbar{{\overline{w}}}

\begin{document}

\title{Twisted heterotic/type I duality}

\author{Kevin Costello}
\address{Perimeter Institute of Theoretical Physics \\ 31 Caroline St N, Waterloo, ON N2L 2Y5, Canada}
\email{kcostello@perimeterinstitute.ca}

\author{Brian R. Williams}
\address{School of Mathematics\\
University of Edinburgh, Edinburgh, UK}
\email{brian.williams@ed.ac.uk}

\date{}

\begin{abstract}
We formulate a twisted version of the conjectured duality between heterotic and type I string theories.  Our formulation relates the chiral part of the heterotic string with a type I topological B-model on a Calabi-Yau five-fold.   We provide a non-trivial check of this duality by showing that  certain infinite-dimensional Lie algebras of global gauge transformations built from each theory are isomorphic. Matching the structure constants of the Lie algebras involves a detailed analysis of one-loop corrections on the type I side.
\end{abstract}

\maketitle

The main object of study of this  paper is the Lie algebra of ghost number $+1$ states in the chiral sector of the heterotic string.
From the string field theory perspective, this Lie algebra is the Lie algebra of gauge symmetries which preserve the vacuum solution to the equations of motion.
We can call this the Lie algebra of global gauge transformations.

Our main calculation shows that this Lie algebra is isomorphic to the Lie algebra of global gauge transformations of the {\em type I topological string} studied in \cite{CLtypeI}.
This is a topological string theory whose target space is a Calabi-Yau $5$-fold and is related to the topological $B$-model in a similar way to how the type I string is related to the physical string.

Our result provides strong evidence of a duality between the chiral part of the heterotic string and the type I topological string, parallel to the standard heterotic-type I duality. 

As such, we provide evidence between two very different classes of world-sheet theories. 
On one side, the type I topological string is basically the topological B-model equipped with an action of the group $O(2)$ \cite{CLtypeI}.   
The other side is a chiral limit of a physical string theory; it is built from a chiral (or holomorphic), {\em not a topological}, $\sigma$-model.  In particular, the type I $B$-model does not have any world-sheet instantons, whereas the chiral heterotic string does. 

This duality also fits with the conjecture made in \cite{CLtypeI} 
relating physical and topological type I strings on a Calabi-Yau $5$-fold.
A better justification for this conjecture would be obtained if we could show that the anti-chiral sector of the heterotic string is trivial in cohomology when we twist by a space -time supersymmetr.
We do not investigate this point here, as space-time supersymmetries live in the Ramond sector making the cohomology computation more difficult, but we hope to return to it in future work.

\subsection*{An outline of the paper}
First, we introduce the chiral sector of the heterotic string on $\CC^{5}$.
We obtain this by first taking an infinite volume limit of the physical heterotic string.  In this limit, the chiral and anti-chiral sectors decouple. 
We then compute the BRST cohomology of local operators in the chiral sector, and finally the Lie bracket on the ghost number $1$ states. 

We then compare with the type I topological string of \cite{CLtypeI}. We find that the Lie algebras are isomorphic. To check this isomorphism, we have to make explicit certain terms in the type I topological string Lagrangian that are forced on us by anomaly cancellation (in \cite{CLtypeI} only the existence of corrections which cancel anomalies was proved, but not all terms were explicitly written down).

\subsection*{Acknowledgements}

We wish to thank Chris Elliott, Owen Gwilliam, Si Li, Natalie Paquette, and Ingmar Saberi for numerous useful conversations on aspects related to this work. 

K.C. is supported by the NSERC Discovery Grant program and by the Perimeter Institute for Theoretical Physics. Research at Perimeter Institute is supported by the Government of Canada through Industry Canada and by the Province of Ontario through the Ministry of Research and Innovation.
B.R.W thanks the University of Edinburgh for their support.

\section{The chiral heterotic string}
Let $X$ be a Calabi-Yau $5$-fold.
Mostly, in the first part of this paper we are concerned with the noncompact case $X = \C^5$, but we will discuss curved backgrounds later on in Section \ref{sec:KS}.
For us, the heterotic string with target $X$ is obtained by coupling the $\so(32)$ Kac--Moody chiral algebra at level $1$ to the $\sigma$-model with target the Calabi--Yau five-fold $X$ with $\cN = (0,2)$ supersymmetry.
In addition, we introduce $\b-\c$ ghosts in the chiral sector, and similar ghosts and super-ghosts in the anti-chiral sector which gauge the $\cN = (0,1)$ super Virasoro algebra, as is standard in string theory \cite{PolchinskiVol2}.

The $\sigma$-model with target $X$ can be written in a first order formulation. When we do this, it consists of a chiral $\beta-\gamma$ system with target $X$, a complex conjugate anti-chiral $\bar{\beta}-\bar{\gamma}$ system, and anti-chiral fermions. The chiral and anti-chiral sectors are coupled by a term which tends to zero in the limit where we send the K\"ahler form on $X$ to infinity while keeping the complex structure fixed.   In this limit, the chiral and anti-chiral parts of the model decouple (except for a topological term which plays no role in our analysis).
For a careful treatment we refer to Appendix A of \cite{GGW}

Similarly, in the limit when we send the K\"ahler form on $X$ to infinity, the entire worldsheet theory of the heterotic string decouples as a product of chiral and anti-chiral sectors.  The point is that the ghost systems enforcing diffeomorphism invariance, as well as the fermions in the anti-chiral sector, were already decoupled even when the K\"ahler form is finite.  See \cite{GWstring} for a detailed discussion of this in the case of the bosonic string.

Thus, the chiral part of the heterotic string consists of the following three pieces:
\begin{enumerate}
	\item The $\beta-\gamma$ system with target the Calabi-Yau manifold $X$.
	\item The level one $\so(32)$ VOA $\mathfrak{so}(32)_1$, which is equivalently the VOA of $32$ free fermions.
	\item The $\b-\c$ ghost system which gauges worldsheet conformal transformations.
\end{enumerate}

\begin{rmk}
If we were to use the $E_8 \times E_8$ heterotic string, one would replace item (2) with the $\mf{e}_8 \oplus \mf{e}_8$ current algebra, again at level $1$.
We will use free-fermions as a model for the fermionic matter system. 
One can repeat our analysis by replacing it with the appropriate Kac--Moody current algebra and the same results will hold.
\end{rmk}

The anti-chiral sector, which will not play a role for us, is a bit more complicated. 
The bosonic fields involve an anti-holomorphic $\Bar{\beta}-\Bar{\gamma}$ system with target $X$.
In addition to fermionic matter fields, the system is coupled to the anti-chiral $\cN=(0,1)$ superconformal ghost system consisting of the $\Bar{b}-\Bar{c}$ ghost system and its superpartner. 

\section{The cohomology of local operators}

We turn our attention to the local operators of the chiral part of the heterotic string on flat space.
The algebra of operators is generated by five chiral fields $\gamma^i$ of spin zero, and five chiral fields $\beta_i$ of spin $1$, $i= 1,\ldots 5$, which satisfy the OPE
\begin{equation}
	\beta_i (0) \gamma^j (z) \;  \simeq \; \frac{1}{z} \delta^j_i.
\end{equation}
We also have $32$ real fermions $\psi_a$, $a = 1,\ldots 32$, which satisfy the OPE
\begin{equation}
	\psi_a (0) \psi_b (z) \; \simeq \; \frac{1}{z} \delta_{a b}  .
\end{equation}
Finally, we have the $\b-\c$ ghosts, where $c$ is of spin $-1$ and ghost number $+1$, $\b$ is of spin $+2$ and ghost number $-1$, and they obey the OPE
\begin{equation}
	\b (0)  \c (z) \; \simeq  \; \frac{1}{z}.
\end{equation}
The BRST current is
\begin{equation}
	J_{BRST} = \sum_i \c \beta_i \partial_z \gamma^i + \sum_a \c \psi_a \partial_z\psi_a + \frac{1}{2}\c \b \partial_z \c  .
\end{equation}
The value of the BRST operator $Q$ applied to a local operator $\mc{O}$ is
\begin{equation}
	Q \mc{O}(0) = \oint_z J_{BRST}(z) \mc{O}(0)
\end{equation}
where we fix a small contour around $0 \in \CC$.

\subsection*{Level matching}
As usual in string theory, this complex is too large and certain states must be removed by imposing a level matching constraint. When there are are chiral and anti-chiral sectors, it is usually implemented by asking that operators are in the kernel of the zero mode of the operator $\b - \br{\b}$.
In our purely chiral setting, we do the same thing,  and only consider states in the kernel of $\b_0$.
Equivalently, we simply remove all states which contain $\partial_z \c$.  Since, in the BRST complex, $\partial_z \c$ enforces invariance under the rotation $L_0 = z \partial_z$, once we remove $\partial_z \c$ we also only consider the subcomplex of states invariant under $L_0$.
We denote by
\begin{equation}
\left(\cS^\bu = \oplus_i \cS^i [-i] \; , \; Q \right)
\end{equation}
the resulting subcomplex of the full BRST complex of the local operators in the chiral sector of the heterotic string.
Here $\cS^i$ denotes those operators of ghost number $i$.

Note that this constraint is simply the standard level matching constraint once we assume that in the right moving sector we are considering the vacuum state.  Only by applying this constraint do we find a sub-sector of the full heterotic string. 

This maneuvre also has a solid mathematical justification. The $\c$-ghost enforce diffeomorphism invariance, here with respect to the holomorphic transformations of an infinitesimal disc $D$.  The Lie algebra of such holomorphic transformations is the Lie algebra of vector fields on an infinitesimal disc, which as a vector space is $\C[[z]] \partial_z$.
Alternating words in the $\c$-ghost and its derivatives are the Lie algebra cochains of $\C[[z]] \partial_z$, and the terms in the BRST operator which add a $\c$-ghost are the differential for the Lie algebra cochain complex of $\C[[z]] \partial_z$ with coefficients in a module built from the remaining operators.

In mathematical terminology, the cochain complex computing Lie algebra cohomology of $\C[[z]] \partial_z$ with coefficients in any module is the \emph{derived} invariants of $\C[[z]] \partial_z$.  
It is important to remember, however, that the holomorphic transformations we are considering form not just a Lie algebra, but a group.\footnote{Strictly speaking, only the transformations which preserve the basepoint form a group, the vector field $\partial_z$ can not be exponentiated.  On a formal disc, all other vector fields can be exponentiated. }  The difference is only seen for the subalgebra generated by $z \partial_z$, which exponentiates to the group $\C^\times$, which is a reductive group.  For reductive groups like this, derived invariants and strict invariants coincide.  Therefore, instead of taking the derived invariants of $\C[[z]]\partial_z$, we should take the \emph{actual} invariants with respect to $\C^\times$, and the derived invariants with respect to the remaining transformations.  This is achieved by the procedure mentioned above, which is known as {\em relative} Lie algebra cohomology.

\subsection*{The chiral heterotic states}

Before computing the BRST cohomology, we explicitly describe the cochain complex $\cS^\bu$ of the level-matched states for the heterotic string.
As we discussed above, the computation is simplified by noting that we are only interested in states that are invariant under $L_0$ (spin zero) and that there is only one operator of negative spin, namely $\c$.

The relevant piece of the BRST complex is
\beqn\label{BRSTseq}
\cdots \to \cS^{0} \xto{Q} \cS^{1} \xto{Q} \cS^{2} \to \cdots
\eeqn

\begin{table}
\begin{tabular}{c|c|c|cccccc}
& {\rm ghost} & $\Pi$ & {\rm spin} \\
\hline
$\partial^\ell \sfc$ & $+1$ & even & $\ell - 1$ \\
$\partial^\ell \sfb $ & $-1$ & even & $\ell + 2$ \\
\hline
$\partial^\ell \gamma^i$ & $0$ & even & $\ell$ \\
$\partial^\ell \beta_j$ & $0$ & even & $\ell+1$ \\
\hline
$\partial^\ell \psi_\alpha$ & $0$ & odd & $\ell+\frac12$ \\
\end{tabular}
\end{table}

First off, it is easy to see that there are no spin zero operators of ghost number less than zero.
Thus, we start at ghost number zero.

\subsubsection*{Ghost number zero}

The only spin zero operators of ghost number zero consist of polynomial functions of the operator $\gamma^i$, $i=1,\ldots, 5$.
Thus, in ghost number zero we have
\begin{equation}
\cS^0 = \CC[\gamma^i] = \Fun(\CC^5) .
\end{equation}

\subsubsection*{Ghost number one}

There are three classes of local operators that are of spin zero and ghost number one:
\begin{itemize}
\item
There is the spin one, ghost number zero linear operator $\partial \gamma^i$ for $i=1,\ldots, 5$.
Then $\partial \gamma^i \sfc$ is a spin zero, ghost number one operator.
\item
There is the spin one, ghost number zero linear operator $\beta_j$, $j=1,\ldots, 5$.
Then $\beta_j \sfc$ is a spin zero, ghost number one operator.
\item
There are the spin one, ghost number zero quadratic operators $\psi_a \psi_b$, $a,b = 1,\ldots, 32$.
Then $\psi_a \psi_b \sfc$ is a spin zero ghost number zero operator.
\end{itemize}

If $\cO$ is any of the spin zero, ghost number one operators in the list above and $f \in \Fun(\CC^5)$ is a polynomial function, then
\begin{equation}
f(\gamma) \, \cO
\end{equation}
is another spin zero, ghost number one operator.
This observation leads us to the following characterization of the spin zero, ghost number one operators.

\begin{lem}\label{lem:gh1}
There is a linear isomorphism
\begin{equation}
\cS^1 \; \cong \; \Omega^1(\CC^5) \oplus {\rm Vect}(\CC^5) \oplus \Fun(\CC^5) \otimes \so(32) .
\end{equation}
\end{lem}
On the right hand side, we mean \emph{algebraic} one-forms, vector fields, and functions, i.e. holomorphic tensors whose coefficients are polynomials.   Given a triple $(\omega, X, \alpha)$ in the right-hand side, we denote the corresponding operators by $\cO_\omega, \cO_X$ and $\cO_\alpha$.
\begin{proof}
Given a holomorphic one-form $\omega = f \; \d w^i \in \Omega^1(\CC^5)$ the corresponding operator is
\begin{equation}
\cO_\omega = \frac12 f (\gamma) \, \partial \gamma^i \sfc .
\end{equation}
The factor of $\frac12$ is a convenient normalization. 

Similarly, a holomorphic vector field $X = X^j(w) \partial_j$  determines the operator
\begin{equation}
\cO_X = X^j(\gamma) \, \beta_j \sfc .
\end{equation}

Finally, a matrix valued function $\alpha = (\alpha_{ab}) \in \Fun(\CC^5) \otimes \so(32)$, where $\alpha_{ab} \in \Fun(\CC^5)$, determines the matrix of operators $\cO_{\alpha} = ((\cO_{\alpha})_{ab})$ where
\begin{equation}
(\cO_{\alpha})_{ab} = \frac12 \alpha_{ab} (\gamma) \psi_a \psi_b \sfc .
\end{equation}
\end{proof}

\subsubsection*{Ghost number two}

The spin zero operators of ghost number two all have the form
\begin{equation}
f(\gamma) \, \sf c \partial^2 \sf c
\end{equation}
where $f \in \Fun(\CC^5)$ is an arbitrary polynomial function.
Thus
\begin{equation}
\cS^{2} \cong \Fun(\CC^5).
\end{equation}

It is easy to check that there are no spin zero operators of ghost number greater than two.

\subsection*{BRST cohomology}

Now let us turn to the computation of the BRST cohomology of the chiral heterotic string.
We will find that the cohomology in concentrated in ghost degrees zero and one.

We have seen that the complex $\cS^\bu$ consisting of operators which are invariant under $L_0$ and do not contain $\partial \sfc$ is concentrated in ghost degrees zero, one, and two.
Of these operators we have enumerated, not all of them are BRST closed.
Moreover, some are BRST exact leading to relations among the operators.

\begin{thm}
The  BRST complex of the chiral sector of the heterotic string is quasi-isomorphic to the cochain complex
\beqn\label{eqn:brstcomplex}
\begin{tikzcd}
\ul{0} & \ul{1} & \ul{2} \\
\Fun(\CC^5) \ar[r, "\partial"] & \Omega^1(\CC^5) & \\
& {\rm Vect}(\CC^5) \ar[r, "{\rm div}"] & \Fun(\CC^5) \\
& \Fun(\CC^5) \otimes \so(32) & .
\end{tikzcd}
\eeqn
Here, $\d$ stands for the holomorphic de Rham operator on $\CC^5$ and ${\rm div}$ is the operator which sends a vector field to its divergence.
\end{thm}
\begin{proof}
Any spin zero, ghost degree zero operator is of the form $f(\gamma)$ where $f \in \Fun(\CC^5)$.
On such operators the BRST operator reads
\begin{equation}
Q (f(\gamma)) = \sum_{i=1}^5 \partial_i f (\gamma) \, \partial \gamma^i \, \sfc \in \cS^1 .
\end{equation}
Notice that this is identical to the ghost number one operator $\cO_{\d f}$ corresponding to the one-form $\d f \in \Omega^1(\CC^5)$.

For any one-form $\omega$, the ghost number one operator $\cO_\omega$ is closed $Q (\cO_\omega) = 0$ for the full BRST differential.

	Now let us consider a vector field $X$, with the operator $\cO_X$ of the form $\cO_X = \beta_i f^i (\gamma) \c$.  The term in the BRST current which will act on $\cO_X$ is $\c \beta_k (\partial \gamma^k)$.  The OPE between $J_{BRST}(z)$ and  $\cO_X$ can involve either one or two Wick contractions of the $\beta$ and $\gamma$ fields. The coefficient of $z^{-1}$ in the single Wick contraction part vanishes, because  all terms involve either $\c \c$ or $\c \partial \c$ and the latter vanishes by level-matching.  The term involve two Wick contractions takes the form
	\begin{equation*} 
		\frac{1}{z^3} \frac{\partial}{\partial \gamma^i} f^i (\gamma) \c(0) \c(z).  
	\end{equation*}
	Expanding $\c(z)$ as $\c(z) = \c(0) + z \partial \c(0) + \frac12 z^2 \partial^2 \c(0) + \cdots $ and retaining the coefficient of $z^{-1}$ in the above expression we find that 
	\begin{equation*} 
		Q( \beta_i f^i(\gamma) \c) =  \frac{\partial}{\partial \gamma} f^i (\gamma) \c \partial^2 \c.  
	\end{equation*}
	In terms of the vector field $X$, the right hand side is the divergence, so that
\begin{equation}
Q (\cO_X) = ({\rm div}(X)) (\gamma) \, \sfc \, \partial^2 \sfc \in \cS^2 .
\end{equation}

Finally, for any $\alpha \in \cO(\CC^5) \otimes \so(32)$ the operator $\cO_{\alpha}$ is BRST closed.
In fact, for any $f \in \Fun(\CC^5)$ and $a,b = 1,\ldots, 32$ we have
\begin{equation}
Q \left(f (\gamma) \psi^a \psi^b \sfc \right) = 0 \in \cS^2 .
\end{equation}

\end{proof}

Thus, we see that the cohomology $H^\bu(\cS)$ is concentrated in degrees zero and one.
In degree zero, only the unit survives $H^0(\cS) \cong \CC$.
In degree one, we obtain
\begin{equation}
H^1 (\cS) \cong \Omega^1 (\CC^5) / {\rm Im} \; \d \oplus {\rm Vect}^{\rm div}(\CC^5) \oplus {\rm Fun}(\CC^5) \otimes \fs \fo (32)
\end{equation}
where ${\rm Vect}^{\rm div}(\CC^5)$ stands for the space of divergence-free vector fields on $\CC^5$.

\section{The algebra of local operators}
\label{heterotic_ope}

The algebra of local operators of a (chiral) string theory does not behave quite like those of a CFT \cite{Getzler1, Getzler2, LZ1, LZ2}. Once we have introduced $\b-\c$ ghosts and taken cohomology, the operator $\partial_z = L_{-1}$ of translation acts cohomologically trivially.  This is simply because worldsheet reparametrization is a gauge symmetry and we are considering gauge invariant quantities.

This immediately implies that for any two BRST closed operators $\mc{O}, \mc{O}'$, the operator product $\mc{O}(0) \mc{O}'(z)$ satisfies
\begin{equation*} 
	\partial_z(\mc{O}(0) \cdot \mc{O}(z) ) = Q ( \dots )
  \end{equation*}
where the $\dots$ stands for some product of local operators.
In particular, at the level of cohomology the OPE has no singularities.
The non-singular part of the OPE gives the cohomology of the space of local operators a commutative product.

For the chiral heterotic string, in BRST cohomology there are only operators of ghost numbers zero and one.
Further, the ghost number zero operator is the identity.
This means that this commutative product is trivial. 

There is an additional structure on the BRST cohomology \cite{LZ1}, namely a Lie bracket of ghost number $-1$:
\begin{equation}
[\cdot , \cdot] : H^i(\cS) \otimes H^j (\cS) \to H^{i+j - 1} (\cS).
\end{equation}
This is defined by the method of topological descent, which we briefly recall.

For any BRST-closed local operator $\mc{O}$ of ghost number $k$ one can build an operator $\cO^{(1)}$ of ghost number $k-1$ such that the descent equation holds
\begin{equation*} 
	\partial_z \cO = Q \cO^{(1)}
\end{equation*}
where $\partial_z$ is the holomorphic derivative.
For this reason, $\cO^{(1)}$ is called the descended operator obtained from $\cO$.
Explicitly,
\begin{equation} 
	\cO^{(1)} = \oint_z \b(z) \cO(0) \d z  
\end{equation}
so that $\cO^{(1)}$ is obtained by differentiating $\cO^{(1)}$ with respect to $c$.
It is easy to verify the descent equation. 

The Lie bracket of ghost number $-1$ is defined by
\begin{equation}
\left[\cO, \cO' \right] =  \oint_{z} \cO^{(1)}(z) \d z \cdot \cO'(0)  .
\end{equation}
Here $\cO$, $\cO'$ are BRST closed operators.
This expression descends to a Lie bracket on BRST cohomology. 
Notice from the equation above that if $\cO$ is of ghost degree $i$ and $\cO'$ is of ghost degree $j$ then the bracket $[\cO, \cO']$ is of degree $i+j-1$.

One can check that this bracket satisfies a graded version of skew-symmetry and the Jacobi identity which endows the shift of cohomology $H^\bu(\cS)[1]$ with the structure of a graded Lie algebra.
Moreover, the bracket is a graded derivation with respect to the graded commutative product on $H^\bu(\cS)$.
Together, these give $H^\bu(\cS)$ the structure of a Gerstenhaber algebra, \cite{LZ1}.

\subsection*{The Lie algebra of the chiral heterotic string}

We return to the chiral sector of the heterotic string.
So far, we have seen that the BRST cohomology of the chiral sector of the heterotic string is concentrated in ghost degrees zero and one.
In degree zero there is just the class of the unit operator.

Summarizing the preceding discussion, we obtain a Lie algebra structure on $H^1(\cS)$ whose bracket we denote by $[-,-]$.
It is characterized as follows.

\begin{prop}
  \label{prop:hetlie}
The Lie bracket on
\begin{equation}
H^1(\cS) \cong {\rm Vect}^{\rm div} (\CC^5) \oplus \Omega^1 (\CC^5) / {\rm Im} (\partial) \oplus \Fun(\CC^5) \otimes \fs \fo (32)
\end{equation}
obtained by topological descent has the following commutators:
\begin{itemize}
	\item for $X = f^i(w^i) \partial_i , Y = g^j (w^j) \partial_j  \in \Vect^{\rm div}(\CC^5)$ the bracket is
\begin{equation}
	[X, Y] =  f^i(\partial_i g^j) \partial_j - g^j (\partial_j f^i) \partial_i    + 2 \left[\partial_j f^i \partial \partial_i g^j\right]   \in {\rm Vect}^{\rm div} (\CC^5) \oplus \Omega^1 (\CC^5) / {\rm Im}(\partial) .
\end{equation}
Note that the first two terms in this expression are the usual commutator of the vector fields $X,Y$.  The remaining term is a one-form, which is only anti-symmetric in $X,Y$ when taken up to exact one forms.

\item for $X \in \Vect^{\rm div}(\CC^5)$ and $\omega \in \Omega^1 (\CC^5) / {\rm Im} (\d)$ the bracket is $[X, \omega] = L_X \omega$, where $L_X (\cdot)$ is the Lie derivative,
\item for $X \in \Vect^{\rm div} (\CC^5)$ and $\alpha = (\alpha_{ab})\in \Fun(\CC^5) \otimes \fs \fo (32)$ the bracket is $[X, \alpha] = (X \cdot \alpha_{ab})$, and
\item for $\alpha, \alpha' \in \Fun(\CC^5)\otimes \fs \fo (32)$ the bracket is
\begin{equation}
[\alpha, \alpha'] = \left(\alpha \alpha' - \alpha' \alpha\right) + \left[\op{tr} \left(\alpha \partial \alpha'\right) \right] \in \Fun(\CC^5) \otimes \so(32) \oplus  \Omega^1 (\CC^5) / {\rm Im} (\d) .
\end{equation}
The first two terms are the usual matrix commutator. The final term is a one-form which is only anti-symmetric in $\alpha,\alpha'$ when taken up to exact one-forms. 
\end{itemize}
\end{prop}

From hereupon we will denote by
\begin{equation}
  \fg_{\rm het} \define H^1(\cS)
\end{equation}
the Lie algebra described in the above proposition.
This is the Lie algebra of global gauge transformations of the chiral sector of the heterotic string.

\begin{proof}
The proof of the proposition is  a direct calculation.

Recall that every vector field $X = X^i(w)\partial_i$ determines the ghost degree one, spin zero local operator $\cO_X = X^i (\gamma) \beta_i \sf c$.
Moreover, we have the following OPE
\begin{equation}
b(0) \cdot \cO_X (z) \simeq \frac{X^i(\gamma) \beta_i} {z} + \cdots.
\end{equation}
Thus, if $Y = Y^j \partial_j$  is another vector field we have
\begin{align}
\left[\cO_X, \cO_Y \right] & = \left(\oint_{z} \left(X^i(\gamma) \beta_i \right) (z) \d z \right) \cdot Y^j (\gamma) \beta_j \sfc \\ & = X^i(\gamma) (\partial_i Y^j)(\gamma) \beta_j \sfc - (\partial_j X^i)(\gamma)Y^j (\gamma) \beta_j \sfc \\ & + (\partial_j X^i) (\gamma) (\partial_k \partial_i Y^j) (\gamma) \partial \gamma^k \sfc .
\end{align}
Each term in the last two lines arises from the $\beta-\gamma$ OPE.
The second to last line follows from a single contraction of the OPE and the last line follows from a double contraction.
We recognize the second to last line as the local operator corresponding to the ordinary commutator of the vector fields $X$ and $Y$.
The last line as the local operator corresponding to the one-form class
\begin{equation}
 2 \left[\partial_j f^i \partial \partial_i g^j\right] = 2\left[\partial_j f^i (w) \partial_k \partial_i g^j(w) \d w^k \right] \in \Omega^1 (\CC^5) / {\rm Im}(\partial)
\end{equation}
as desired.

Next, suppose $[\omega] \in \Omega^1 (\CC^5) / {\rm Im} \; \d$ and pick a representing one-form $\omega = g_j \d w^j$.
Then
\begin{align*}
[\cO_X, \cO_\omega] & = \left(\oint_{z} \left(f^i(\gamma) \beta_i \right) (z) \d z \right) \cdot g_j (\gamma) \partial \gamma^j \sfc \\
& = f^i(\gamma) (\partial_i g_j)(\gamma) \partial \gamma^j \sfc + (\partial_j f^i)(\gamma) g_i (\gamma) \partial \gamma^j \sfc .
\end{align*}
Here, the first term uses $\beta-\gamma$ OPE arising from $g^j(\gamma)$ in the one-form component and the second term uses the $\beta-\gamma$ OPE arising from $\partial \gamma^j$ in the one-form component.
The right-hand side is precisely the operator corresponding to the class $[L_X \omega]$. 
Notice that if we choose a different representative one form $\omega'$ then $[\cO_X, \cO_\omega] - [\cO_X, \cO_{\omega'}]$ is $\d$-exact. 

Next, if $\alpha = (\alpha_{ab}) \in \Fun(\CC^5) \otimes \so(32)$ then for each $a,b$ we have
\begin{align*}
[\cO_X, (\cO_\alpha)_{ab}] & = \frac12 \left(\oint_{z} \left(f^i(\gamma) \beta_i \right) (z) \d z \right) \cdot \alpha_{ab} (\gamma) \psi_{a} \psi_{b} \sf c \\ & = \frac12 (\partial_i \alpha_{ab})(\gamma) \psi_{a} \psi_{b} \sfc
\end{align*}
where, again, we only utilize the $\beta-\gamma$ OPE.
We recognize the second line as the local operator corresponding to $(ab)$ entry of the matrix $X \cdot \alpha$, as desired.

Finally, we observe the following OPE involving the operator $(\cO_{\alpha})_{ab}$
\begin{equation}
\sfb (0) \cdot (\cO_{\alpha})_{ab}(z) \simeq \frac{\frac12 \alpha_{ab} (\gamma) \psi_a\psi_b}{z} + \cdots .
\end{equation}
If $\alpha' = (\alpha'_{cd})$ is another matrix valued function we have
\begin{align*}
[(\cO_{\alpha})_{ab}, (\cO_{\alpha'})_{cd}] &= \left(\oint_z \left(\alpha_{ab} (\gamma) \psi_a \psi_b\right) (z) \d z \right) \cdot \alpha'_{cd} (\gamma) \psi_c \psi_d \sfc \\
& = \frac14 (\alpha_{ab} \alpha'_{cd})(\gamma) \left(\delta_{ad} \psi_{b} \psi_{c} - \delta_{ac} \psi_b \psi_d + \delta_{bc} \psi_{a} \psi_{d} - \delta_{bd} \psi_a \psi_c\right) \sfc \\ & + \frac14 \left(\delta_{ac} \delta_{bd} + \delta_{ad}\delta_{bc}\right) (\alpha_{ab} \partial_k \alpha'_{ba})(\gamma) \partial \gamma^k \sfc  .
\end{align*}
Each term in the last two lines arises from the $\psi-\psi$ OPE.
The second to last line follows from a single contraction of the OPE and the last line follows from a double contraction.
We recognize the second to last line as the local operator corresponding to the commutator of matrices $[\alpha, \alpha']$.
The last line is the local operator corresponding to the one-form class
\begin{equation}
\left[ {\rm tr}(\alpha \partial \alpha')\right] = \left[{\rm tr}(\alpha \partial_k \alpha') \d w^k \right] \in \Omega^1(\CC^5) / {\rm Im}(\partial) .
\end{equation}
\end{proof}

By the holomorphic Poincar\'e lemma, the quotient $\Omega^1(\CC^5) / {\rm Im}(\partial)$ is isomorphic to the space of close two-forms $\Omega^{2}_{\rm cl}(\CC^5)$. 
The identification takes a class $[\omega]$ to the two-form $\partial \omega$, which, of course, does not depend on the representative. 

If we replace this quotient space by the space of closed two-forms, the relevant brackets in the Lie algebra now read
\begin{align*}
[X, Y] & =  f^i(\partial_i g^j) \partial_j - g^j (\partial_j f^i) \partial_i   + 2 \partial (\partial_j f^i) \wedge \partial ( \partial_i g^j) \\
[\alpha, \alpha'] & = \left(\alpha \alpha' - \alpha' \alpha\right) + \op{tr} \left(\partial \alpha \wedge \partial \alpha'\right) .
\end{align*}

\section{Type $I$ -- heterotic duality}

\label{sec:KS}

So far, we have determined the Lie algebra of ghost number one states of the chiral part of the heterotic string. Our main theorem identifies this with the Lie algebra of global gauge symmetries of another string theory: the type $I$ topological $B$-model string introduced in \cite{CLtypeI}.

From the worldsheet perspective, the type I $B$-model is the un-oriented version of the ordinary topological $B$-model. Thus, instead of gauging by oriented world-sheet diffeomorphisms, we gauge by all world-sheet diffeomorphisms. This means that the space of states of the type I topological $B$-model is the $\Z/2$ fixed points of the states of the ordinary $B$-model, where the $\Z/2$ acts by an orientation reversing symmetry of the worldhseet.

As in the physical type I string, this theory is only anomaly free when we also introduce certain open string fields.  On a Calabi-Yau $5$-fold\footnote{Topological strings on Calabi-Yau manifolds of dimension other than $3$ or not much studied in the physics literature, but they make perfect sense.  The only issue is that there is a ``ghost number anomaly''.  This means that the theory makes sense when ghost number is treated modulo $2$, or alternatively when we treat the string coupling constant as having ghost number $3-d$, where $d$ is the complex dimension of the target.}, the open string fields live in a rank $32$ bundle, and introduce an  ${\rm SO}(32)$ gauge theory on the target.

We are interested in the space-time theory of the type I topological string with target a Calabi-Yau $5$-fold $X$, normally taken to be $\C^5$.  This was studied in \cite{CLtypeI}, and was found to have fields which are the $\Z/2$ fixed points of the fields of Kodaira-Spencer theory \cite{BCOV}, coupled to an ${\rm SO}(32)$ holomorphic Chern-Simons gauge theory.  It was conjectured that this space-time theory describes the holomorphic twist in the sense of \cite{CLsugra} of type I supergravity, and that more generally the type I topological $B$-model we consider is equivalent to a twist of the physical type I string.

Given the expected duality \cite{WittenDynamics, PolchinskiWitten,HoravaWitten} between the type I string and the heterotic string, we would hope for a similar duality between the type I topological string and an appropriate part of the heterotic string.

We conjecture that the appropriate part of the heterotic string is simply the chiral part\footnote{It is not completely clear how to define the chiral part of the heterotic string by itself at loop level. Loop amplitudes should be defined by a contour integral over a middle-dimensional cycle in the moduli of Riemann surfaces with marked points. The choice of cycle plays the role of the string vertices, and as usual the string vertices should satisfy an appropriate quantum master equation.  We leave the investigation of appropriate integration cycles to future work.  }, leading to the conjectured duality:
\begin{conjecture}
	There is an equivalence between the type I topological string on a Calabi-Yau $5$-fold $X$ and the chiral part of the $SO(32)$ heterotic string with target $X$.
\end{conjecture}
This conjecture must be non-perturbative. For instance, $D1$ branes in the type I topological string should correspond to world-sheet instantons in the chiral part of the heterotic string.  Of course, such non-perturbative conjectures are very difficult to check.  In \cite{CLtypeI} an initial check was provided:  it was noted that the theory on a single $D1$ brane in the type I topological string matches the chiral part of an embedded heterotic string (where the $\b-\c$ ghosts are gauged fixed by the chosen embedding).

In this paper we provide a complimentary check.  We have computed the Lie algebra of ghost number $1$ states of the chiral part of the heterotic string. Now, we will show that this Lie algebra is isomorphic to a similar Lie algebra built from the type I topological string.  

On the type I side, this Lie algebra can be computed by either space-time or world-sheet methods. We will focus on the space-time approach (although it would be illuminating to carefully perform the world-sheet computation which we hope to return to in later work).
From the space-time side, the Lie algebra is that of gauge symmetries of the space-time theory on $\C^5$, which preserve the zero field configuration.

\subsection{Type $I$ Kodaira-Spencer theory}

Let $X$ be a Calabi--Yau five-fold with holomorphic volume form $\Omega$.
In the original formulation of Kodaira--Spencer theory in \cite{BCOV}, one describes the fields by polyvector fields on $X$ which are divergence free with respect to $\Omega$.
In other words, on looks at the locus
\begin{equation}
  \op{Ker} \partial_{\Omega} \subset \PV^{\bu,\bu} (X) .
\end{equation}
Here, $\PV^{i,j}(X)$ stands for Dolbeault forms of type $(0,j)$ valued in $i$th exterior power of the holomorphic tangent bundle $\wedge^i T_X^{1,0}$.
If we work on a Calabi-Yau three-fold, the fields in $\PV^{i,j}$ have ghost number $i+j-2$.
On a Calabi-Yau $5$-fold, the ghost number anomaly means that the ghost number is $i+j$ modulo $2$.

\begin{rmk}
In the formulation of \cite{CLbcov1} the condition that a polyvector field be closed for the divergence operator is implemented homologically, by introducing extra fields and extra terms in the BRST operator.

One introduces a formal variable $u$ of even cohomological degree.
The extended space of fields of (type II) Kodaira--Spencer theory on a Calabi--Yau five-fold $X$ is
	\begin{equation} 
	\bigoplus_{\ell \geq 0}  u^\ell \PV^{k, \bu} (X) 	 
	\end{equation}
where elements of $u^\ell \PV^{i,j}$ are of parity $i+j$ modulo $2$.   

The linear BRST operator is $\dbar + u \partial_{\Omega}$ where $\partial_{\Omega}$ is the holomorphic divergence operator with respect to the holomorphic volume form $\Omega$.
The extra fields, so-called ``gravitational descendants'', are those fields which have a nontrivial dependence on the parameter $u$.
From the world-sheet perspective the descendent field $u$ has a very natural interpretation discussed in \cite{CLtypeI}, it arises from ``large'' world-sheet diffeomorphisms of a circle.

For simplicity we will stick with the presentation of \cite{BCOV} here without introducing these additional fields.
\end{rmk}

The fields of type I Kodaira-Spencer theory are the $\Z/2$ fixed points of the fields of ordinary Kodaira-Spencer theory  under a $\Z/2$ which acts as $(-1)^{k+1}$ on $\PV^{k,l}$.
This $\Z/2$ action arises from the action of an orientation-reversing symmetry of the worldsheet on the space of states, see \cite{CLtypeI} for details.  The $\Z/2$ fixed points are thus the fields of the un-oriented closed-string field theory:
\begin{equation} 
	\op{Ker} \partial \subset \PV^{3,\bu} \oplus \PV^{1,\bu}.
\end{equation}
We let $\mu$ denote the field in $\PV^{1,\bu}$ and $\eta$ that in $\PV^{3,\bu}$.  The component of $\mu$ in $\PV^{1,k}$ is of ghost number $1+k$ modulo $2$, as is the component of $\eta$ in $\PV^{3,k}$. The Lagrangian for just the closed-string sector is
\begin{equation} 
	\int \left[\Omega \vee \left( \eta \wedge \left( \partial^{-1} \dbar \mu + \tfrac{1}{2} \mu \wedge \mu \right) \right) \right] \wedge \Omega .
\end{equation}
Note that if $\mu$ is in $\PV^{1,k}$ then $\partial^{-1} \dbar \mu$ is in $\PV^{2,k+1}$, and $\mu \wedge \mu$ indicates the natural wedge product on polyvector fields. 

Varying $\eta$ leads to the equation of motion
\begin{equation}
  \dbar \mu + \frac12 [\mu,\mu]_{\rm NS} = 0  
\end{equation}
Specializing $\mu = \mu^{1,1}$ to be in $\PV^{1,1}$, for example, this equation tells us that $\mu^{1,1}$ is a Beltrami differential defining an integrable deformation of complex structure on $X$.

It will be important to understand the  BV anti-bracket. Expand $\mu$ and $\eta$ into components $\mu^i_{\br{I}} \partial_i \d \zbar^I$, $\eta^{ijk}_{\br{J}} \partial_i\partial_j \partial_k \d \zbar^{\br{J}}$ where $\br{I}$, $\br{J}$ are multi-indices.  We view these components as operators in the theory. Then, the BV anti-bracket is 
\begin{equation} 
	\{\mu^i_{\br{I}}(z), \eta^{jkl}_{\br{J}}(w)\} = \eps^{ijklm} (\partial_m \delta^{(5)}_{z = w})\eps_{\br{I} \br{J}} . \label{BV_closed} 
\end{equation}

The type I topological string can be coupled to holomorphic Chern-Simons theory on the Calabi-Yau $X$. The fundamental field of holomorphic Chern-Simons theory is 
\begin{equation} 
	A \in \Pi \Omega^{0,\bu}(X,\mf{g})
\end{equation}
where $\mf{g}$ is the gauge Lie algebra and $(0,k)$ forms have parity $k-1$ modulo $2$.  This field $A$ includes all the ghosts and anti-fields of the theory; because of the ghost number anomaly, it is not possible to write the theory in a way using only fields of ghost number $0$ and gauge symmetries.  

The Lagrangian is
\begin{equation} 
	\int \Omega \wedge CS(A) = \int \Omega \wedge \op{Tr} \left(\tfrac{1}{2} A \dbar  A + \tfrac{1}{6} A[A,A]\right). 
\end{equation}

At tree-level, the coupling to the fields of Kodaira-Spencer theory (as studied in \cite{CLtypeI}) is of the form
\begin{equation}
	\frac12 \int (\mu\vee \Omega ) \op{tr}(A \partial A) = \frac12 \int \mu^i \op{tr} (A \partial_i A) .    \label{OC1a}
  \end{equation}
Here, $\op{tr}(\cdot)$ stands for the trace in the fundamental (vector) representation.
The BV anti-bracket for the Kodaira-Spencer fields is, if we expand $A$ as $A_{\br{I}} \d \zbar^I$, 
\begin{equation} 
	\{A_{\br{I}}(z), A_{\br{J}}(w)\} = \delta^{(5)}_{z = w} \eps_{\br{I} \br{J}}.
\end{equation}

\subsection{The Lie algebra of global symmetries, first pass}
In any field theory, the parity shift of the space of fields has an $L_\infty$ structure so that the equations of motion of the theory are the Maurer-Cartan equations of the $L_\infty$ algebra \footnote{The existence of the $L_\infty$ structure is a tautology: this is a feature of any polynomial equation, not just of the equations of motion of a field theory.  See \cite{DAGX}. This $L_\infty$ formalism is used extensively in, for instance, \cite{CG2}.}.

Concretely, one builds this $L_\infty$ algebras as follows.  Expand the action functional $S$ as a sum of quadratic, cubic, etc.\ terms:
\begin{equation} 
	S = S^{(2)} + S^{(3)} + \dots 
\end{equation}
The differential $l_1$ of the $L_\infty$ algebra is given by the BV anti-bracket with $S^{(2)}$, and the Lie bracket $l_2$ is given by the BV anti-bracket with $S^{(3)}$, and so on.

The cohomology of this $L_\infty$ algebra is a $\Z/2$ graded Lie algebra, which we call here the Lie algebra of global gauge transformations\footnote{The terminology is justified by noting that when there is no ghost number anomaly, so that we have a $\Z$ instead of $\Z/2$ graded Lie algebra, the cohomology in degree $0$ is the global gauge transformations.}.  Our goal is to compute this for Kodaira-Spencer theory coupled to holomorphic Chern-Simons theory.  This Lie algebra 
is insensitive to the quartic and higher terms in the Lagrangian we start with. 

For type I Kodaira-Spencer theory, this $L_\infty$ algebra is a $\Z/2$ graded dg Lie algebra. The differential is given by the linear BRST operator, which is simply the $\dbar$ operator. The bracket given by the Nijenhuis--Schouten bracket of polyvector fields.
This is a bilinear operator of the form
\begin{equation}
[\cdot, \cdot]_{\rm NS} \colon \PV^{i,\bu} (X) \times \PV^{j,\bu}(X) \to \PV^{i+j-1,\bu}(X) .
\end{equation}
The Maurer--Cartan equations for this dg Lie algebra is, by definition, the equations of motion for type I Kodaira--Spencer theory on $X$.

Taking $X = \C^5$, we find that the cohomology of this Lie algebra lives in even degrees, and consists of
\begin{enumerate} 
	\item $X \in \op{Vect}^{\rm div}(\C^5)$, the divergence free holomorphic vector fields. 	
  \item $\omega \in \Omega^2_{\rm cl}(\C^5)$, a $\partial$-closed holomorphic two-form (we identify $\Omega^{2,k}$ with $\op{PV}^{3,k}$ here using the Calabi--Yau structure).
    Equivalently, we can view $\omega$ as a holomorphic $1$-form, taken up to exact holomorphic one-forms.
\end{enumerate}
The Lie bracket when restricted to $\op{Vect}^{\rm div}(\C^5)$ is simply the standard Lie bracket of vector fields.  The bracket of two elements of $\Omega^1(\C^5)/ \op{Im} \partial$ is zero, and the bracket of $X$ with $\omega$ is given by the Lie deriative.

Let us now include the holomorphic Chern-Simons fields.   
The $\dbar$-cohomology of $\Omega^{0,\bu}(\C^5, \mf{g})$ is simply  $\op{Fun}(\C^5, \mf{g})$, the Lie algebra of holomorphic (or polynomial) maps from $\C^5$ to $\mf{g}$. 
The coupling to the Kodaira--Spencer field $\mu$ as in Equation \eqref{OC1a} produces the following two terms in the Lie bracket:
\begin{enumerate} 
	\item The Lie bracket of $X \in \op{Vect}^{\rm div}(\C^5)$ with $\alpha \in \op{Fun}(\C^5) \otimes \mf{g}$ is given by differentiation $X \cdot \alpha$.
  \item The Lie bracket of two elements $\alpha,\alpha' \in \op{Fun}(\C^5) \otimes \mf{g}$ acquires an extra term in $\Omega^2_{\rm cl}(\C^5)$ given by $\op{Tr}(\partial \alpha \partial \alpha')$.
    Lifting this to a $1$-form modulo exact one-forms, it is the class of $\op{Tr}(\alpha \partial \alpha')$.
\end{enumerate}
(The second term is derived as follows. We consider the coupling $\int (\mu \vee \Omega) {\rm tr}(A \partial A)$ as in \eqref{OC1a}.  The BV anti-bracket of this with $\eta$, which we view as an element of $\Omega^2_{\rm cl}$, involves $\partial$, leading to the expression written above).

So far, we have found that the Lie algebra of global symmetries of the type I topological string is \emph{almost} the same as what we found from the chiral 
part of the heterotic string.   Indeed, if we specialize to $\mf{g} = \mf{so}(32)$, which is forced on us by anomaly cancellation, we find that the two Lie algebras have the same underlying vector space, and almost the same Lie bracket.

The only difference is that in the heterotic string calculation, we found a term whereby the Lie bracket of two vector fields $X^i \partial_i$, $Y^i \partial_i$ gives rise to a one-form:
\begin{equation} 
	  \partial_j X^i \partial_k \partial_i Y^j \d w^k    
\end{equation}
Since this one-form is taken up to exact one forms, we can view it as the closed two-form
\begin{equation} 
	 \partial_l \partial_j X^i \partial_k \partial_i Y^j \d w^l \d w^k. 
\end{equation}

This discrepancy can be corrected by adding a term proportional to 
\begin{equation}
	\int \mu^k \partial_i  \mu^j \partial_k \partial_j \mu^i \label{C4}  
\end{equation}
to the Lagrangian of type I Kodaira-Spencer theory. (Here we include only holomorphic indices; anti-holomorphic indices are all anti-symmetrized).  

Since it is a cubic term, it modifies the Lie bracket.  To figure out how it modifies the Lie bracket, we can study the equations of motion in the presence of this extra term.  To quadratic order in the fields, the equations of motion obtained by varying $\mu$ are 
\begin{equation} 
	\dbar \eta^{ijk}  + [\mu,\eta]_{\rm NS}^{ijk} +    \eps^{ijklm}	\partial_l  \partial_i  \mu^j \partial_m \partial_j \mu^i = 0 .
\end{equation}
The last term on the left hand side is the new term. For this to be the Maurer-Cartan equation of the new dg Lie algebra, we need an additional term in the Lie bracket which is precisely that we found studying the chiral heterotic string.

The main question that remains to answer is, where does this extra term in the type I topological string Lagrangian come from?   We will find it is forced on us to cancel a certain one-loop anomaly.  

\begin{rmk}
  From the worldsheet perspective, we believe that this term comes from an analysis of the scattering of three closed-string states with worldsheet $\RR \PP^2$ with three punctures.  We have not, however, performed this computation in detail.
\end{rmk}


\section{Anomaly cancellation in the Type I topological string}

Classically, type I Kodaira--Spencer theory can be coupled to holomorphic Chern--Simons theory for any Lie algebra.  However, holomorphic Chern-Simons theory has a one-loop gauge anomaly associated to a hexagon diagram \cite{GShex}.
The topological-string version of the Green-Schwarz mechanism \cite{CLtypeI} means that this can be cancelled when the gauge Lie algebra of holomorphic Chern-Simons theory is $\mf{so}(32)$.

This was shown as follows.  First, it was shown that the anomaly coming from the pure gauge hexagon anomaly in holomorphic Chern-Simons theory can be cancelled by a tree-level diagram involving the exchange of two closed string fields, as long as the gauge Lie algebra is $\mf{so}(32)$ (we recall this argument below).
Then, by a cohomological argument, it was shown that any further anomalies {\em can be} cancelled by the addition of counter-terms.
This step of the argument in \cite{CLtypeI} is implicit, in that no further anomalies or counter-terms were computed.  

In this paper we will need to make this implicit step explicit, by computing another anomaly and a correction to the action which cancels it.
This correction modifies the Lie algebra of global symmetries we have computed above, yielding it isomorphic to what we find in the chiral heterotic string.

\subsection*{One-loop anomalies}

For theories like holomorphic Chern--Simons theory and Kodaira--Spencer theory, one-loop anomalies are entirely explicit to characterize.
By general results of \cite{CLbcov1} and \cite{BWhol} the one-loop anomaly for any holomorphic field theory is given as the sum over wheel graphs with a fixed number of vertices.
For theories on $\CC^{5}$ one restricts to wheels with six vertices, matching with the usual take on anomalies for theories defined on ten-dimensional spacetimes.

The vertices of such wheels can be labeled by the two interactions we have introduced: the classical Chern--Simons coupling $\frac16 \int \Tr(A[A,A])$ and the tree-level open-closed string coupling $\frac12 \int \mu^i \tr(A \partial_i A)$. 
Since our gauge Lie algebra is $\mathfrak{so}(32)$ only wheels labeled by an even number of Chern--Simons vertices will be nonzero. 

Thus, the total one-loop anomaly is of the form
\[
\Theta = \Theta_0 + \Theta_2 + \Theta_4
\]
where $\Theta_k$ is the wheel with $k$ vertices labeled by the open-closed coupling, see Figure \ref{fig:anomaly}.

\begin{rmk}
We have not included a potential term in the anomaly that depends just on $\mu$.
This is a hexagon wheel whose vertices are all labeled by the $\mu A A$ vertex in \eqref{OC1a}, and hence gives rise to a local functional which is sixth order in $\mu$.
We will not pay this term much attention, but remark that it is compensated for by the gravitational anomaly which is built from purely closed string couplings and the closed string propagator.
The cancellation of these two terms forces the dimension of the gauge Lie algebra to be $496$.
\end{rmk}
In each of the graphs, the unmarked vertices are labeled by the holomorphic Chern--Simons interaction $\frac16 \int A [A,A]$.
The vertices involving $\mu$, displayed with a bullet $\bu$, are labeled by
\begin{equation}
 \frac12 \int \mu^i \op{tr} (A \partial_i A)
\end{equation}
which is the tree-level open-closed coupling that we introduced in (\ref{OC1a}).
Here, as usual, we drop anti-holomorphic indices, which are always contracted with $\eps_{\br{j}_1 \dots \br{j}_5}$ and we have dropped the holomorphic volume form for simplicity in the second line.

The first anomaly $\Theta_0$ is the standard pure gauge anomaly of holomorphic Chern--Simons theory.
For any gauge Lie algebra, this anomaly takes the form of the local functional
\begin{equation}
\Theta_0 = \frac{1}{6!} \int \eps^{ijklm} {\rm Tr} (A \partial_i A \partial_j A \partial_k A \partial_l A \partial_m A)
\end{equation}
where ${\rm Tr}(\cdot)$ denotes the trace in the adjoint representation.
The constant $1/6!$ can be deduced from an explicit calculation of the graph integral.\footnote{We do not include factors of $(2 \pi)^5$ that should appear in this analysis.}
The one-loop anomaly arises from the graph in Figure \ref{fig:anomaly} with all vertices labeled by the holomorphic Chern--Simons action $\frac16 \int {\rm Tr}(A [A, A])$ and all internal edges labeled by the holomorphic Chern--Simons propagator.

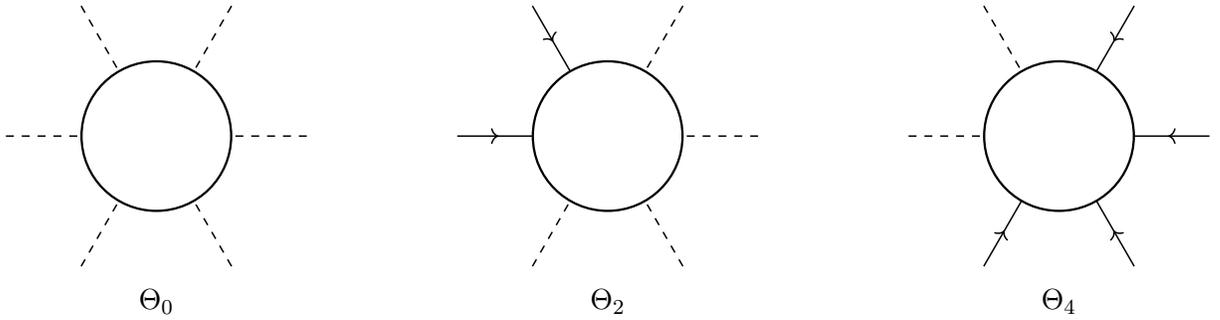
\begin{figure}
\begin{center}
\begin{tikzpicture}[line width=.2mm, scale=1]


		\draw[fill=black] (-6,0) circle (1cm);
		\draw[fill=white] (-6,0) circle (0.99cm);

        \coordinate (a) at (-6,0);

		\draw[dashed] (a)+(0:2) --+ (0:1);
		\draw[dashed] (a)+(60:2) --+ (60:1);
		\draw[dashed]  (a)+(120:2) --+ (120:1);
		\draw[dashed]  (a)+(180:2) --+ (180:1);
		\draw[dashed] (a)+(240:2) --+ (240:1);
		\draw[dashed] (a)+(300:2) --+ (300:1);
		
		\node at (-6,-2.2) {$\Theta_0$};
		
		\draw[fill=black] (0,0) circle (1cm);
		\draw[fill=white] (0,0) circle (0.99cm);

        		\coordinate (b) at (0,0);

		\draw[dashed] (b)+(0:2) --+ (0:1);
		\draw[dashed] (b)+(60:2) --+ (60:1);
		\draw[fermion]  (b)+(120:2) --+ (120:1);
		\draw[fermion]  (b)+(180:2) --+ (180:1);
		\draw[dashed] (b)+(240:2) --+ (240:1);
		\draw[dashed] (b)+(300:2) --+ (300:1);

		\draw[fill=black] (6,0) circle (1cm);
		\draw[fill=white] (6,0) circle (0.99cm);
		
		\node at (0,-2.2) {$\Theta_2$};
 		
		\coordinate (c) at (6,0);

		\draw[fermion] (c)+(0:2) --+ (0:1);
		\draw[fermion] (c)+(60:2) --+ (60:1);
		\draw[dashed]  (c)+(120:2) --+ (120:1);
		\draw[dashed]  (c)+(180:2) --+ (180:1);
		\draw[fermion] (c)+(240:2) --+ (240:1);
		\draw[fermion] (c)+(300:2) --+ (300:1);
		
		\node at (6,-2.2) {$\Theta_4$};


	    	\clip (-2,0) circle (1cm);
\end{tikzpicture}
\caption{Potential open-closed anomalies. 
Solid lines label closed-string fields $\mu$, dashed lines label open-string fields $A$.}
\label{fig:anomaly}
\end{center}
\end{figure}

Following \cite{CLtypeI}, one mechanism to trivialize the pure gauge anomaly $\Theta_0$ is to modify the open-closed string coupling to a term involving the closed-string field 
\[
\eta = \eta^{ijk} \partial_i \wedge \partial_j \wedge \partial_k \in {\rm PV}^{3,\bu}(\CC^3) .
\] 
The modification is
\begin{equation}
 \label{OCI}
I = \frac12 \int \mu^i \op{tr} (A \partial_i A) + b_1 \int \eta^{jkl} \op{tr} (A \partial_j A \partial_k A \partial_l A) .
\end{equation}
for some constant $b_1$, which is determined by matching the coefficients with the anomaly.
Notice that the first term above is precisely the tree-level open-closed coupling \eqref{OC1a}.

The closed-string BV anti-bracket of the two terms coupling open and closed string fields is proportional to
\begin{equation}
	\left\{ \frac12 \int \mu^i \op{tr}(A \partial_iA), b_1  \int \eta^{jkl} \op{tr} (A \partial_j A \partial_k A \partial_l A) \right\} = \frac{b_1}{2}\int \eps^{ijklm} \op{tr} (A \partial_i A) \op{tr} ( \partial_j A \partial_k A \partial_l A \partial_m A) .
\end{equation}
(The extra derivative on the right hand side arises because the closed string BV anti-bracket \eqref{BV_closed} involves a derivative).

Thus, using $I$ we can cancel the pure gauge anomaly via the BV bracket $\frac12 \{I , I\} = \Theta_0$ provided
\begin{equation}
  {\rm Tr} (X^{6}) \, \propto \, {\rm tr} (X^{2}) {\rm tr}(X^{4})
\end{equation}
for all $X \in \fg$.
On the left hand side appears the trace in the adjoint representation and on the right hand side appears the trace in the fundamental representation.
For $\fg = \so(n)$ this is only true when $n=32$, in which case ${\rm Tr}(X^6) = 15 \, {\rm tr}(X^2) {\rm tr}(X^4)$.
In order for the constants of proportionality to work out one must therefore have $b_1 = 1/4!$. 


The remaining terms $\Theta_2$, $\Theta_4$ in the anomaly are built from both the Chern--Simons interaction {\em and} the open-closed string coupling. 
These are shadows of gravitational and mixed gauge-gravitational anomalies at the level of the topological string.
In principle, we know this anomaly to be cohomologically trivial since we know a quantization exists.
In the remainder of this section we will construct an trivialization of the additional terms in the one-loop anomaly.


The anomalies $\Theta_2, \Theta_4$ depend on both the open and closed string fields:
\begin{itemize}
  \item[$\Theta_2$:]
    This term is quartic in the holomorphic Chern--Simons field $A$ and quadratic in the closed string field $\mu$.
    All internal edges are labeled by the holomorphic
    Chern--Simons propagator.
  \item [$\Theta_4$:] This term is quadratic in the holomorphic Chern--Simons field $A$ and quadratic in the closed string field $\mu$.
    All internal edges are labeled by the holomorphic Chern--Simons propagator.
\end{itemize}

Explicitly, we will only be interested in trivializing the anomaly $\Theta_2$.
The graph calculation of this anomaly is summarized by the following proposition.
The proof of this proposition is rather technical and can be found in Appendix \ref{appx:graph}.
The key ingredient is the of the precise form of the holomorphic Chern--Simons propagator on $\CC^{5}$ which we recall in the appendix.


\begin{prop}
\label{prop:anomaly}
The local contribution $\Theta_{2}$ of the weight of the graph (a) in Figure \ref{fig:anomaly} is
\begin{equation}
	\Theta_2 = \frac{1}{12} \frac{1}{2 \cdot 4!} \int \partial_i \mu^j \partial_{k_1} \partial_j \mu^i \op{Tr}\left(\partial_{k_2} A \partial_{k_3} A \partial_{k_4} A \partial_{k_5}  A\right) \varepsilon^{k_1 \dots k_5} . \label{mixed_anomaly_2}
\end{equation}
\end{prop}

\begin{rmk}
While we will not need the explicit form of the piece of the anomaly which is quartic in the closed string field, we remark that its local contribution $\Theta_{4}$ is a linear combination of functionals of the form
\[
\int \left(\partial_i \mu^j \partial_{k_1} \partial_j \mu^\ell \partial_{k_2} \partial_\ell \mu^m \partial_{k_3} \partial_m \mu^i \right) \op{Tr}(\partial_{k_4} A \partial_{k_5} A) \varepsilon^{k_1 \dots k_5}
\]
and
\[
\int \left(\partial_i \mu^j \partial_{k_1} \partial_j \mu^i \partial_{k_2} \partial_\ell \mu^m \partial_{k_3} \partial_m \mu^\ell \right) \op{Tr}(\partial_{k_4} A \partial_{k_5} A) \varepsilon^{k_1 \dots k_5}
\]
\end{rmk}

In turns out that just in the case of the pure gauge anomaly, one can introduce an additional terms in the Lagrangian which trivialize the mixed anomaly $\Theta_2$ (and $\Theta_4$).

We wrap all the necessary couplings into a new functional of the form
\begin{equation}
\begin{array}{cclll}
J & = & b_2 \int \mu^k \partial_i \mu^j \partial_k \partial_j \mu^i + b_3 \int \eta^{k \ell m}  \partial_i \mu^j \partial_k \partial_j \mu^i \tr (\partial_\ell A \partial_m A) \\ & + & \int \eta^{k_1 k_2 k_3} \left(b_4 \partial_i \mu^j \partial_{k_1} \partial_j \mu^\ell \partial_{k_2} \partial_\ell \mu^m \partial_{k_3} \partial_m \mu^i + b_5 \partial_i \mu^j \partial_{k_1} \partial_j \mu^i \partial_{k_2} \partial_\ell \mu^m \partial_{k_3} \partial_m \mu^\ell\right) .
\label{OCJ}
\end{array}
\end{equation}
where $b_2, b_3, b_4, b_5$ are two yet to be determined coefficients.
The functional in the second line will not play an essential role for us.
Tuning the coefficients $b_4,b_5$ amount to trivializing the term in the anomaly $\Theta_4$. 

  We show that the total coupling $I+J$, where $I$ was defined in \eqref{OCI}, trivializes the total anomaly 
  \[
    \frac12 \{I + J, I + J\} \mod \hbar^2 = \Theta_0 + \Theta_2 + \Theta_4.
  \]
  We expand the right-hand side as $\frac12 \{I,I\} + \{I,J\} + \frac12 \{J, J\} \mod \hbar^2$.
  We have already seen that the first term trivializes the pure gauge anomaly $\Theta_{0}$.

The anomaly $\Theta_2$ is canceled by terms in the BV anti-bracket between $I$ (see Equation \eqref{OCI}) and $J$ (see Equation \eqref{OCJ}):
  \begin{itemize}
      \item The bracket of the second term in \eqref{OCI} with the first term in \eqref{OCJ}:
\begin{equation}
  \label{eqn:j1}
  \begin{array}{lllll}
	& \left\{ \frac{1}{4!} \int \eta^{lmn} \op{tr} (A \partial_l A \partial_m A \partial_n A) \; , \; b_2 \int \mu^k \partial_i  \mu^j \partial_{k} \partial_j \mu^i \right\} = \\ & \frac{3b_2}{4!} \int \left( \partial_i  \mu^j\partial_{k_1} \partial_{j} \mu^i \right)  \op{tr} \left(\partial_{k_2} A \partial_{k_3} A \partial_{k_4} A \partial_{k_5}  A\right) \varepsilon^{k_1 \dots k_5}
	\end{array}
  \end{equation}
    \item The bracket of the first term in \eqref{OCI} and the second term in \eqref{OCJ}:
      \begin{equation}
  \label{eqn:j2}
  \begin{array}{lllll}
	&
  \left\{ \frac12 \int \mu^i \tr(A \partial_i A) \; , \; b_3 \int \eta^{k \ell m} \partial_i \mu^j \partial_k \partial_j \mu^i \tr(\partial_\ell A \partial_m A) \right\} = \\ 
  & \frac{b_3}{2} \int \left( \partial_i  \mu^j\partial_{k_1} \partial_{j} \mu^i \right)  \op{tr} \left(\partial_{k_2} A \partial_{k_3} A \right) \op{tr}\left( \partial_{k_4} A \partial_{k_5}  A\right) \varepsilon^{k_1 \dots k_5}
  \end{array}
    \end{equation}
   

  \end{itemize}

Notice that on the right hand side of \eqref{eqn:j1}--\eqref{eqn:j2} we see the trace in the fundamental, not the adjoint, representation of $\so(32)$.
The sum of \eqref{eqn:j1} and \eqref{eqn:j2} is
\begin{equation}\label{eqn:j4}
  \int \left(\partial_i  \mu^j\partial_{k_1} \partial_{j} \mu^i \right) \left[\frac{3 b_2}{4!} \op{tr} \left(\partial_{k_2} A \partial_{k_3} A \partial_{k_4} A \partial_{k_5}  A\right) + \frac{b_3}{2} \op{tr} \left(\partial_{k_2} A \partial_{k_3} A \right) \op{tr}\left( \partial_{k_4} A \partial_{k_5}  A\right) \right] \varepsilon^{k_1 \dots k_5} .
\end{equation}

By requiring this term matches with the mixed anomaly $\Theta_{2}$ we will obtain relations among the various coefficients. 
Just as in the case of the pure gauge anomaly, we will use a certain trace identity for $\mf{so}(32)$.
This time it involves polynomials of order four:
\[
  \op{Tr}(X^4) = 24 \op{tr}(X^4) + 3 [\op{tr}(X^2)]^2 .
\]
On the left-hand side appears the trace in the adjoint representation, and on the right-hand side appears the trace in the fundamental representation.
Using this, we observe that for \eqref{eqn:j4} to be proportional to the mixed anomaly $\Theta_2$ we must have $b_2 = 32 b_3$. 
With this relation, \eqref{eqn:j4} reduces to
\begin{equation}
  \label{eqn:j5}
  \frac{b_3}{6} \int \left(\partial_i  \mu^j\partial_{k_1} \partial_{j} \mu^i \right) \op{Tr} \left(\partial_{k_2} A \partial_{k_3} A \partial_{k_4} A \partial_{k_5}  A\right) \varepsilon^{k_1 \cdots k_5}
\end{equation}
which is equal to $\Theta_2$ if and only if 
\[
b_3 = 6 \cdot  \frac{1}{12} \frac{1}{2 \cdot 4!} 
\]
This implies $\boxed{b_2 = 1/3}$.

We conclude that the term $\frac13 \int \mu^k \partial_i  \mu^j \partial_k \partial_j \mu^i$ must be added to the action in order to cancel the one-loop anomaly. 
This term in the action gives rise to the term in the Lie bracket between two vector fields $[X^i \partial_i , Y^j \partial_j]$ given by $2 \partial (\partial_j X^i \partial_i Y^j)$ as in Proposition~\ref{prop:hetlie}.

\appendix

\section{The mixed one-loop anomaly}\label{appx:graph}

In this appendix we supply the proof of Proposition \ref{prop:anomaly}.
The proof relies on the existence of a special `holomorphic gauge' which exists for holomorphic theories on $\CC^{n}$, see Proposition 4.4 \cite{BWhol}. 
This gauge has been used by Costello--Li \cite{CLbcov1} to characterize the {\em pure gauge} anomaly of the open-closed topological string. 
We recall the main ingredients of this result and then apply it to the case at hand in the type I topological string.

First off, we clarify what we mean by a holomorphic theory on $\CC^{n}$.
For more details we refer to Chapter 9 of \cite{CG2} or \cite{LiVertex,BWhol}.
This is a theory whose fields take the form
\begin{equation}
  \Omega^{0,\bu}(\CC^n, V)
\end{equation}
where $V$ is some graded holomorphic vector bundle on $\CC^n$.
The linear BRST operator is of the form
\begin{equation}
  \dbar + D
\end{equation}
where $\dbar$ is the $\dbar$-operator for the holomorphic vector bundle $V$ and $D$ is some translation invariant holomorphic differential operator acting on $V$.
The operator $D$ is required to have cohomological degree $+1$ and satisfy the equation $(\dbar + D)^2 = 0$.

In the BV formalism we require $V$ to be equipped with a non-degenerate graded skew-symmetric holomorphic pairing of degree $n-1$.
The pairing induces the BV anti-bracket on the space of fields $\Omega^{0,\bu}(\CC^n, V)$, which uses the obvious holomorphic form $\d^n z$.

So far we have just discussed aspects of {\em free} holomorphic field theories on $\CC^n$.
An interacting holomorphic field theory is, by definition, a deformation of a free theory by local functionals of the form $\int_{\CC^n} L^{\rm hol} \; \d^n z$ where $L^{\rm hol}$ is a {\em holomorphic Lagrangian density}.
This means that $L^{\rm hol}$ can be written as a sum of Lagrangian densities of the form
\begin{equation}
  D_1 (\varphi) \cdots D_k (\varphi) 
 \end{equation}
where the $D_i$'s are all holomorphic differential operators.

As an example, consider holomorphic Chern--Simons theory on $\CC^5$.
Here, the bundle $V$ is taken to be the trivial bundle with fiber the Lie algebra $\fg[1]$.
Here, we are only working in the $\ZZ/2$-graded sense, so $\fg[1]$ is understood as vector space concentrated in purely odd degree.
If $\fg$ is equipped with a non-degenerate symmetric invariant bilinear form $(-,-)_\fg$ then $\fg[1]$ is equipped with an odd invariant bilinear form that we denote by the same symbol.
We obtain an odd pairing on fields by the expression $\int_{\CC^5} (X,Y)_\fg \d^5 z$ for $X,Y \in \fg[1]$.

In the case $\fg = \so(n)$ we use the trace pairing to yield the BV pairing on the fields of holomorphic Chern--Simons theory $\Omega^{0,\bu}(\CC^5 , \fg)[1]$ by the formula
\begin{equation} \label{eqn:tracepairing}
  \int_{\CC^5} \op{tr} (A A') \d^5 z .
\end{equation}
Holomorphic Chern--Simons is an interacting holomorphic field theory with interaction given by the standard formula $\frac16 \int A [A,A]$.

We now turn to the issue of quantization for holomorphic theories on $\CC^{n}$.
An appealing aspect of this class of theories is that there exists a choice of a regularized propagator such that the theory is \emph{finite} at one loop. 
This means  that all one-loop counterterms are identically zero.

By construction, the trace pairing \eqref{eqn:tracepairing} is given by a distributional kernel. 
A particular resolution for this kernel is given by a solution to the heat equation on $\CC^5$ with respect to the standard Hodge Laplacian.
For $t > 0$, this resolved BV kernel takes the explicit form
\beqn\label{heat}
  K_t (z,w) = \frac{1}{(4 \pi t)^5} e^{-|z-w|^2 / 4L} \prod_{i=1}^5 (\d \zbar^i - \d \wbar^i) \otimes {\rm Id}_{\fg} \in \Omega^{0,\bu}(\CC^5 \times \CC^5) \otimes (\fg \otimes \fg) .
\eeqn
The regularized propagator is
\beqn\label{propagator}
  P_{\epsilon < L} (z, w) = \int_{t = \epsilon}^L \d t \frac{1}{(4 \pi t)^5} \sum_{i=1}^5 \left(\frac{\zbar^i - \wbar^i}{4t}\right) e^{-|z-w|^2/4t} \prod_{j\ne i} (\d \zbar^i - \d \wbar^j) \otimes {\rm Id}_\fg .
\eeqn

\begin{rmk}
The regularized propagator is obtained from the heat kernel for the Hodge Laplacian associated a choice of a Hermitian metric on $\CC^n$ by the formula $P_{\epsilon<L} = \int_{t=\epsilon}^L \dbar^* K_t \d t$.
Here, $\dbar^*$ is the adjoint of the $\dbar$-operator associated to the fixed Hermitian metric.
\end{rmk}

In our situation, the most important aspect of holomorphic field theories is that their one-loop anomalies to quantization can be characterized explicitly.
As a regularized quantity, the anomaly is a functional $\Theta[L]$ depending on the length scale $L$ and measures the failure of the renormalized action to satisfy the quantum master equation.
It is completely determined by its local behavior which is the local functional obtained by taking the $L \to 0$ limit of $\Theta[L]$.

The local contribution to the one-loop anomaly for any holomorphic field theory on $\CC^n$ is given as the $L \to 0$ limit of the Feynman weight of the graph that is a wheel with $n+1$ vertices.

The vertices are labeled by the classical interaction.
A single edge is labeled by the heat kernel $K_\epsilon$ and the remaining $n$ edges are labeled by the propagator $P_{\epsilon < L}$.
The weight is given by performing the necessary contractions and taking the $\epsilon \to 0$ limit.
The local contribution of the anomaly is the $L \to 0$ limit of the resulting quantity.

With the requisite notation in place, we now turn towards the proof of Proposition \ref{prop:anomaly}.
This proposition describes a term in the one-loop anomaly of holomorphic Chern--Simons theory on $\CC^5$ coupled to type I Kodaira--Spencer theory, see Figure \ref{fig:anomaly3}.

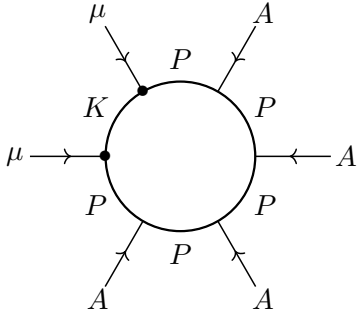
\begin{figure}
\begin{center}
\begin{tikzpicture}[line width=.2mm, scale=1]



		\draw[fill=black] (0,0) circle (1cm);
		\draw[fill=white] (0,0) circle (0.99cm);

		\draw[fermion](0:2) -- (0:1);
		\draw[fermion](60:2) -- (60:1);
		\draw[fermion](120:2) -- (120:1);
		\draw[fermion](180:2) -- (180:1);
		\draw[fermion](240:2) -- (240:1);
		\draw[fermion](300:2) -- (300:1);

		\node at (0:2.2) {$A$};
		\node at (60:2.2) {$A$};
		\node at (120:2.2) {$\mu$};
		\node at (180:2.2) {$\mu$};
		\node at (240:2.2) {$A$};
		\node at (300:2.2) {$A$};
		\node at (120:1) {$\bullet$};
		\node at (180:1) {$\bullet$};

        \node at (30:1.3) {$P$};
        \node at (90:1.3) {$P$};
        \node at (150:1.3) {$K$};
        \node at (210:1.3) {$P$};
        \node at (270:1.3) {$P$};
        \node at (330:1.3) {$P$};

	    	\clip (0,0) circle (1cm);
\end{tikzpicture}
\caption{The mixed anomaly. Here, the label $P$ on the edges stands for the propagator $P_{\epsilon <L}$ and $K$ stands for the heat kernel $K_\epsilon$.}
\label{fig:anomaly3}
\end{center}
\end{figure}

By the above discussion this term arises from a wheel graph with six vertices where two of the vertices are labeled by the $\mu AA$ interaction which is  

\begin{equation}
	= \frac12 \int \mu^i \tr (A\partial_i A) .
\end{equation}
The remaining vertices are labeled by the holomorphic Chern--Simons interaction\\$\frac16 \int \Tr(A [A,A])$.

Five of the edges are labeled by the holomorphic Chern--Simons propagator $P_{\epsilon <L}$ as in (\ref{propagator}).
The remaining edge is labeled by the holomorphic Chern--Simons BV kernel (\ref{heat}).

For nonzero $0 < \epsilon < L$, the weight of this diagram will be given as the integral of a certain $\epsilon$ and $L$-dependent differential form on $(\CC^5)^6$.
By the finiteness of the one-loop effective action for holomorphic theories the $\epsilon \to 0$ limit of this weight is guaranteed to exist.
The $L \to 0$ limit is the local functional representing the anomaly.

Label the coordinate of each vertex as $z_{(1)},\ldots, z_{(6)}$.
To simplify the weight computation, we use the ``center of mass'' coordinates where
\begin{align*}
w_{(1)} & = z_{(2)} - z_{(1)} \\
 & \vdots \\
w_{(4)} & = z_{(5)} - z_{(4)} \\
w_{(5)} & = z_{(5)} .
\end{align*}

The $\epsilon$,$L$-dependent weight of the diagram in Figure \ref{fig:anomaly3} is given by the following expression 
\begin{multline}
\sum_{j_1,j_2} \int_{(\CC^5)^6} \mu^{j_1} (z_{(1)}) \mu^{j_2} (z_{(2)}) \op{Tr} \left(A(z_{(3)}) A (z_{(4)}) A(z_{(5)}) A(z_{(6)})\right) \\
\times \int_{T_1, \ldots, T_5 = \epsilon}^L \d^5 T  \sum_{i_1, \ldots, i_5} \ep^{i_1\cdots i_5} \left(\frac{\wbar_{(1)}^{i_1} \wbar_{(1)}^{j_1}}{(4 T_1)^2} \right)\left(\frac{\wbar_{(2)}^{i_2} \wbar_{(2)}^{j_2}}{(4 T_2)^2} \right) \left(\frac{\wbar_{(3)}^{i_3}}{4 T_3} \right) \left(\frac{\wbar_{(4)}^{i_4}}{4 T_4} \right) \left(\frac{\wbar_{(5)}^{i_5}}{4 T_5} \right) \\ \times {\rm exp} \left(-\sum_{k=1}^5 \frac{|w_{(k)}|^2}{4 T_k} - \frac{1}{4\epsilon} \left|\sum_{k=1}^5 w_{(k)} \right|^2 \right)
\end{multline}

One proceeds by performing the Gaussian integral in the variables $w_{(1)}, \ldots, w_{(4)}$. 
The first nonzero term in the Wick expansion of the Gaussian integration is a nonzero multiple of
\begin{align*}
& \int_{\CC^5} \partial_{j_2} \mu^{j_1} \partial_{i_1} \partial_{j_1} \mu^{j_2} \op{Tr} \left(\partial_{i_2} A \partial_{i_3} A \partial_{i_4} A  \partial_{i_5} A\right) F(\epsilon, L)
\end{align*}
where $F(\epsilon, L)$ is some function which is independent of the fields. 
One can check that the $\epsilon \to 0$ limit of this expression exists, and the further $L \to 0$ limit is a nonzero multiple of the local functional of Proposition \ref{prop:anomaly}.

We do not perform a careful computation of the precise factor obtained by taking the $\epsilon \to 0, L \to 0$ limit of the function $F(\epsilon , L)$. 
Rather, we deduce the coefficient in front of the anomaly counterterm using the local index theorem. 

\subsection*{Anomaly polynomials}



We can use index theory to match the local anomaly of type I Kodaira--Spencer theory coupled to holomorphic Chern--Simons theory that we found in the previous section.
Suppose $X$ is a Calabi--Yau manifold of complex dimension five and $P$ is a principal $\SO(32)$ bundle on $X$. 
For $n = 2,3\ldots$, let
\[
\ch_n (T_X) \in H^{2n}(X) 
\]
denote the Chern characters of the holomorphic tangent bundle. 
Also, let 
\[
\ch_{n} (P) \in H^{2n} (X)
\]
be the Chern characters of $P$. 

The one-loop anomaly to quantizing the coupled theory on $X$ is given by 
\[
[\Td(X) \ch(T_X - \CC)]_{12} + [\Td(X) \ch(P)]_{12} \in H^{12} (X) .
\]

We record the following identities for the components of the Todd class:
\begin{itemize}
\item $\Td_0 = 1$. 
\item $\Td_2 = \frac{1}{12} \ch_2(T_X)$.
\item $\Td_4 = \frac{1}{240} \left(\frac{5}{6} \ch_2(T_X)^2 + 2 \ch_4(T_X)\right)$. 
\item $\Td_6 = - \frac{1}{2016} \left( \frac{7}{36} \ch_2(T_X)^3 + \frac{7}{5} \ch_2(T_X) \ch_4(T_X) + 8 \ch_6(T_X) \right)$. 
\end{itemize}

The piece of the anomaly involving $\ch_4(P)$ is
\[
\frac{1}{12} \ch_2(T_X) \ch_4(P) .
\]
The local representative for this term is
\[
\frac{1}{12} \left(\frac12 \op{Tr}(R^2)\right) \left(\frac{1}{4!} \op{Tr}(F^4) \right) = \frac{1}{12} \frac{1}{2 \cdot 4!} \op{Tr}(R^2) \op{Tr}(F^4) .
\]
These factors are the factors which appear in Proposition \ref{prop:anomaly}. 


\printbibliography

\end{document}